\documentclass[10pt,conference]{IEEEtran}
\usepackage[mathcal]{euscript}
\usepackage{amssymb,amsmath,amsthm}
\usepackage{graphicx}
\usepackage{times}
\usepackage{graphicx}
\usepackage{amsfonts}
\usepackage{color}
\usepackage[T1]{fontenc}
\usepackage{mathtools}
\usepackage{epstopdf}
\usepackage{mathtools}
\usepackage{bbm}
\usepackage{booktabs}

\newtheorem{remark}{Remark}

\newtheorem{theo}{Theorem}

\newtheorem{assumption}{Assumption}
\newtheorem{corollary}{Corollary}

\ifCLASSINFOpdf
\else
\fi

\begin{document}

\title{Performance Analysis of Multi-Service Oriented Multiple Access Under General Channel Correlation}

\author{Nassar Ksairi and M\'erouane Debbah\\
Mathematical and Algorithmic Sciences Lab,
France Research Center,\\
Huawei Technologies France SASU, 92100 Boulogne-Billancourt, France.\\
Emails:
\{nassar.ksairi, merouane.debbah\}@huawei.com
}

\IEEEoverridecommandlockouts

\maketitle

\begin{abstract}
In this article, we provide both analytical and numerical performance analysis of multi-service oriented multiple access (MOMA), a recently proposed non-orthogonal multiple-access scheme for scenarios with a massive number of concurrent connections originating from separate service classes with diverse quality-of-service (QoS) profiles and running on both handheld terminals and Internet-of-Things (IoT) devices. MOMA is based on both class dependent hierarchical-spreading transmission scheme and per-class reception structure. The performance analysis presented in this article is based on realistic channel and signal models for scenarios where the base station is equipped with a large number of antennas. It provides asymptotically exact approximations of the ergodic rates achievable by MOMA.
\end{abstract}

\IEEEpeerreviewmaketitle


\section{Introduction}

Recent radio access proposals for machine type communications (MTC), whether intended as an {\it in-band} part of future Long-Term Evolution (LTE) and fifth-generation (5G) networks~\cite{tti_bundling} or as an {\it out-of-band} standalone system such as LoRa$^{\textrm{TM}}$ and SIGFOX$^{\textrm{TM}}$~\cite{smart_city}, offer improved IoT deployment capabilities than previous cellular technologies. However, they suffer from limitations in their support for at least one of the following qualities: flexibility in resource assignment, efficiency in resource utilization and differential treatment of diverse MTC QoS and traffic profiles. Indeed, while all of the existing in-band IoT compatible radio access solutions are based on reserving a region in the resource grid for MTC, they have limited flexibility in resource allocation within that region. For instance, resource assignment in eMTC \cite{tti_bundling} is limited to choosing for each MTC transmission one (or several) resource blocks (RBs) out of a total of up to only 6 RBs. The higher resource granularity of narrow-band IoT (NB-IoT)~\cite{tti_bundling} comes at the price of using a different numerology in the MTC region than the one used in the rest of the resource grid  by means of narrow-band frequency division multiple access (FDMA). Achieving this numerology requires in practice the insertion of wasteful guard bands. Moreover, orthogonal multiple-access schemes (such as FDMA in NB-IoT and SIGFOX$^{\textrm{TM}}$ and quasi-orthogonal code division multiple access (CDMA) in LoRa$^{\textrm{TM}}$) are not the best in terms of scalability \cite{cdma_with_pc} with respect to increasing device densities. Furthermore, existing proposals achieve MTC coverage enhancement by means of signal repetition thanks to transmission time interval (TTI) bundling \cite{tti_bundling}. However, this transmission technique is far from optimal in terms of resource utilization efficiency \cite{coverage_enhancement}. Another issue is that MTC transmissions do not all have the same QoS and traffic characteristics. Indeed, while some IoT emissions will be mainly small packets of data from simple sensors, other transmissions will carry video streaming from unmanned aerial vehicles and other IP-enabled cameras \cite{uav_cps}. Moreover, some of the services running on handheld devices, such as messaging, have traffic characteristics and data rate requirements that resemble IoT services. One step towards supporting diverse traffic and QoS profiles across both the IoT and the handheld categories is the {\it prioritized random access} scheme (see \cite{multi_service_ra1} and \cite{multi_service_ra2}) which introduces different treatment in the random-access step for five classes of services/traffic profiles, namely broadband communications, high-priority MTC, low-priority MTC, scheduled MTC and emergency services. However, this support is not yet extended to multiple access.

MOMA is a in-band multiple-access scheme that has been conceived for scenarios with multiple service classes. It is based on a novel code domain multiplexing method that we dubbed {\it service dependent hierarchical spreading}~(\cite{eucnc_2016,commag_2017}) and which guarantees support for both inter-class (quasi) orthogonality and different degrees of resource overloading across different classes. Thanks to its use of the code domain, MOMA offers more flexibility in assigning resources to different service classes and to different connections within each class without being subject to the constraints of narrow-band FDMA and while offering higher efficiency in resource utilization. These advantages were partially analyzed in \cite{eucnc_2016} under an ideal array correlation model and only one linear combining method at the base station. In this article, we present a more thorough performance analysis of MOMA under more general channel correlation model and receiver structure using both theoretical tools from large random matrix theory and simulations.

\paragraph*{Notations}
The $N\times N$ identity matrix is denoted by $\mathbf{I}_N$ while notation
$\mathbf{1}_{N\times M}$ stands for the $N\times M$ matrix with all its
entries set to one. The spectral norm of matrix $\mathbf{M}$ is denoted as
$\|\mathbf{M}\|$ while $\mathrm{tr}(\mathbf{M})$ designates the trace of
$\mathbf{M}$.


\section{System Model}
\label{sec:model}

Consider a wireless system consisting of one base station (BS) equipped with $M$ antennas (indexed using $m\in\{1,\ldots,M\}$) serving $K$ single-antenna user terminals (UTs) that have \emph{uplink} data to transmit and which are grouped into $C\geq2$ service classes. Each service class $c\in\{1,\ldots,C\}$ is characterized with a target uplink data rate $r_c$ that is common to all the class members. The $C$ classes are indexed in a way that their respective target data rates satisfy $r_1>r_2>\ldots>r_C$.
In the sequel, we use $\mathcal{K}_c\subset\{1,2,\ldots,K\}$ to designate the indexes of the users of the $c$-th service class. We also define $K_c\stackrel{\mathrm{def.}}{=}\left|\mathcal{K}_c\right|$. Signal transmission is done using orthogonal frequency division multiplexing (OFDM) with $N_{\mathrm{FFT}}$ subcarriers, $N_{\mathrm{CP}}$-long cyclic prefix and a total bandwidth of $W$ Hz. Assume that the OFDM resource grid is structured into RBs and focus on a transmission that is taking place on $N$ RBs which are adjacent in the frequency, the time or in both the frequency and the time domains as shown in Figure~\ref{fig:CE}. Note that RB adjacency in the time domain could be the result of TTI bundling that will be supported for several transmission scenarios in next-generation cellular systems \cite{tti_bundling}.
The wireless link from UT $k$ to BS antenna $m$ is a multi-path channel with $L\geq1$ dominant paths that remain constant within the duration of an OFDM symbol but which might vary from one symbol to another. Denote by $t\in\{1,\ldots,T\}$ the index of the $T\geq N$ OFDM symbols that fall within the $N$ RBs and let $h_{k,m,t}(.)$ designate the continuous-time impulse response of this channel during the $t$-th OFDM symbol, then
\begin{equation}
 \label{eq:cont_imp_resp}
h_{k,m,t}(\tau)=\sum_{l=0}^{L-1}\alpha_{k,m,l,t}\delta(\tau-\tau_l),\:\forall
\tau\in\mathbb{R},
\end{equation}
where $\tau_l$ is the delay of the $l$-th path and $\alpha_{k,m,l,t}$ is its complex-valued gain. Path gains $\left\{\alpha_{k,m,l,t}\right\}_{l\in\{0\cdots L-1\}}$ are modeled, as in~\cite{mimo_ofdm}, as mutually independent zero-mean random variables with variances $\sigma_l^2\stackrel{\mathrm{def.}}{=} \mathbb{E}\left[\left|\alpha_{k,m,l,t}\right|^2\right]$ that satisfy $\sum_{l=0}^{L-1}\sigma_l^2=1$ and
\begin{align}
 \label{eq:time_corr}
\mathbb{E}\left[\alpha_{k,m,l,t}\alpha_{k,m,l,s}^*\right]&=\sigma_l^2
J_0\left(2\pi f_k^D\left(N_{\mathrm{FFT}}+N_{\mathrm{CP}}\right)T_s(t-s)\right)
\nonumber\\
&\stackrel{\mathrm{def.}}{=}\sigma_l^2r_k^{\alpha}(t-s),
\end{align}
where $f_k^D$ is the maximum Doppler frequency shift and where $T_s\stackrel{\mathrm{def.}}{=}1/W$. Now, define $\mathbf{h}_{k,l,t}\stackrel{\mathrm{def.}}{=} \left[h_{k,1,l,t}\: \cdots\: h_{k,M,l,t}\right]^{\mathrm{T}}$, $\boldsymbol{\alpha}_{k,l,t}\stackrel{\mathrm{def.}}{=}$ $\left[\alpha_{k,1,l,t}\: \cdots\: \alpha_{k,M,l,t}\right]^{\mathrm{T}}$ and assume $\boldsymbol{\alpha}_{k,l,t}\sim \mathcal{CN}\left(\mathbf{0},\mathbf{R}_{k,l}^{\alpha}\right)$, where $\mathbf{R}_{k,l}^{\alpha}\stackrel{\mathrm{def.}}{=}\mathbb{E}\left[\boldsymbol{\alpha}_{k,l,t}\boldsymbol{\alpha}_{k,l,t}^{\mathrm{H}}\right]$. Finally, assume as in \cite{ST_analysis} and \cite{3gpp_kronecker} that the correlation between the two vectors $\boldsymbol{\alpha}_{k,l,t_1}$ and $\boldsymbol{\alpha}_{k,l,t_2}$ where $t_1\neq t_2$ is separable into time domain and space domain multiplicative terms. More precisely, $\mathbb{E}\left[\boldsymbol{\alpha}_{k,l,t_1}\boldsymbol{\alpha}_{k,l,t_2}^{\mathrm{H}}\right]=r_k^{\alpha}(t_1-t_2)\mathbf{R}_{k,l}^{\alpha}$. From \eqref{eq:cont_imp_resp}, the frequency domain channel coefficient between UT $k$ and the $m$-th BS antenna at subcarrier $n$ during OFDM symbol $t$, denoted as $H_{k,m,t}$, is given by
\begin{equation}
 \label{eq:freq_ch}
H_{k,m,t,n}=\sum_{l=0}^{L-1}\alpha_{k,m,l,t}
e^{-2\pi\imath\frac{\tau_l n}{N_{\mathrm{FFT}} T_s}}.
\end{equation}
Now, define $\mathbf{H}_{k,t,n}\stackrel{\mathrm{def.}}{=}\left[H_{k,1,t,n} \cdots H_{k,M,t,n}\right]^{\mathrm{T}}\in\mathbb{C}^{M\times1}$ and denote by $x_{k,t,n}$ the sample transmitted on the $n$-th subcarrier during the $t$-th OFDM symbol. The corresponding $M\times1$ vector of samples received at the BS array then writes as
\begin{equation}
 \label{eq:sig_model}
\mathbf{y}_{t,n}=\sum_{k=1}^{K}\sqrt{g_k}\mathbf{H}_{k,t,n}x_{k,t,n}+
\mathbf{z}_{t,n},
\end{equation}
where $\mathbf{z}_{t,n}$ is a $M\times1$ vector of $\mathcal{CN}(0,\sigma^2)$ noise samples and $g_k$ is the large-scale fading factor.
\begin{figure}
 \centering
 \includegraphics[width=1.00\hsize]{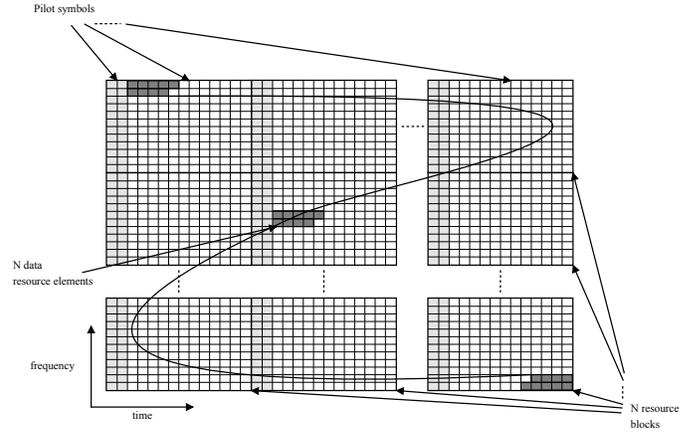}
 \caption{An example of a resource block assignment and of a possible associated
pilot symbol pattern}
 \label{fig:CE}
\end{figure}
Next, we show how samples $x_{k,t,n}$ are constructed from the data symbols of UT $k$ depending on the service class to which UT $k$ belongs.


\section{Multi-Service Oriented Multiple Access}

\subsection{MOMA Transmitter}
\label{sec:moma_mapping}

Let $\mathbf{U}$ be a $N\times N$ orthogonal code matrix, e.g. a Walsh-Hadamard matrix, a discrete Fourier transform (DFT) matrix, etc. The $u$-th column of $\mathbf{U}$ is normalized such that
\begin{equation}
 \label{eq:power_u}
\left\|\left[\mathbf{U}\right]_u\right\|^2=1,\:\forall u\in\{1,\ldots,N\}.
\end{equation}
Next, split $\mathbf{U}$ into $C$ sub-matrices $\{\mathbf{U}_c\}_{c=1\cdots C}$ with the $c$-th sub-matrix having dimensions $N\times N_c$ that satisfy
\begin{equation}
 \label{eq:scenario_of_interest}
\frac{K_C}{N_C}>\cdots
\frac{K_2}{N_2}>\frac{K_1}{N_1}.
\end{equation}
Now, let $\mathbf{W}_c\in\mathbb{C}^{N_c\times K_c}$ be a set of predefined $C$ matrices, each having mutually distinct columns. In MOMA, the $N$-long spreading code $\mathbf{c}_k$ for $k\in\mathcal{K}_c$ is defined as
\begin{equation}
 \label{eq:spreading_code_k_c}
\mathbf{c}_k=\left[\mathbf{U}_c\mathbf{W}_c\right]_{i_k}\nonumber
=\sum_{u=1}^{N_c}\left[\mathbf{w}_{k}\right]_u\left[\mathbf{U}_c\right]_u
\end{equation}
where $\mathbf{w}_{k}\stackrel{\mathrm{def.}}{=}\left[\mathbf{W}_c\right]_{i_k}$ is the $i_k$-th column of $\mathbf{W}_c$ (for some $i_k\in\{1,\ldots,K_c\}$) and where $\left[\mathbf{w}_{k}\right]_u$ is the $u$-th entry of $\mathbf{w}_k$. Let $s_k$ designate one of the unit-variance data symbols from UT $k$. Any such data symbol $s_k$ is then multiplied with vector $\mathbf{c}_k$ resulting in $N$ spread samples $\mathbf{c}_k s_k$. To limit as much as possible the effect of channel gain variations on the desired properties of the code, we map the samples resulting from spreading one data symbol to a group of $N$ neighboring resource elements (REs). Let $\mathcal{N}$ be the group of REs used to transmit the spread samples $\mathbf{c}_k s_k$ and denote by $x_{k,t,n}$ the scaled spread sample transmitted on $(t,n)\in\mathcal{N}$, then
\begin{equation}
\label{eq:tx_symb}
x_{k,t,n}=\sqrt{P_c}\left[\mathbf{c}_k\right]_{(t,n)}s_k,
\end{equation}
where $\left[\mathbf{c}_k\right]_{(t,n)}$ is the entry of $\mathbf{c}_k$ mapped to $(t,n)$.
Vector $\mathbf{w}_{k}$ will help the BS distinguish the signal of UT $k$ from the interfering signals transmitted by the other UTs that belong to the same service class as $k$. In principle, $\mathbf{W}_c$ could be any $N_c\times K_c$ complex-valued matrix with entries chosen that the following average power constraint is satisfied
\begin{equation}
 \label{eq:avg_tx_p_constraint}
\mathbb{E}\left[\left|x_{k,t,n}\right|^2\right]=P_c\quad (k\in\mathcal{K}_c)
\end{equation}
\begin{remark}
The following shorter-term average transmit power constraint might be preferred in order to ease the requirements on the transmit power amplifiers:
\begin{equation}
\label{eq:power_w_c}
\left\|\mathbf{w}_k\right\|^2=1,\:\forall k\in\{1,\ldots,K_c\}.
\end{equation}
Indeed, if \eqref{eq:power_u} and \eqref{eq:power_w_c} are respected, then $\left\|\mathbf{c}_k\right\|^2=1$ and the short-term constraint $\frac{1}{N}\sum_{(t,n)\in\mathcal{N}}\left|x_{k,t,n}\right|^2=P_c$
replaces the long-term constraint in \eqref{eq:avg_tx_p_constraint}. One practical way to obtain matrices $\mathbf{W}_c$ that satisfy \eqref{eq:power_w_c} is to use {\it pseudo-noise} (PN) generators \cite{lte_book} with different initializations to generate the different columns of $\mathbf{W}_c$.
However, since PN sequences are binary, the number of different $N$-long sequences we can get will be upper bounded by $2^N$.
A more systematic method (left for future work) is to select $\left\{\mathbf{w}_k\right\}_{k\in\mathcal{K}_c}$ as a collection of $K_c$ points on the surface of a unit-radius $K_c$-dimensional complex-valued sphere with a minimal pairwise angle that is as large as possible. This problem can be solved in advance for different configurations of $(N_c,K_c)$~\cite{cube_split} and the outcome of this offline optimization can be stored in look-up tables.
\end{remark}


\subsection{MOMA Receiver} 
\label{sec:moma_rec}

\subsubsection{Channel Estimation}
\label{sec:CE}

Assume that $N_p$ REs are reserved for pilot transmission within the $N$ RBs and denote by $\mathcal{P}$ the set of positions of these REs. Let $\mathcal{P}_k\subset\mathcal{P}$ denote the subset of pilot resource elements used in estimating the channel from UT $k\in\mathcal{K}_c$. Let $p_{k,t_i^p,n_i^p}$ designate the pilot symbol transmitted by UT~$k$ at $(t_i^p,n_i^p)\in\mathcal{P}_k$ with power $P_c$. Pilot sequences from different UTs are assumed here to be multiplexed using some combination of frequency division multiplexing (FDM) and time division multiplexing (TDM) resulting in $N_k^p=N_p/K$.
\begin{remark}
To reduce pilot overhead, the same service dependent hierarchical spreading principle that is used for data transmission in MOMA can in principle be applied to pilots. Performance analysis of such a non-orthogonal pilot multiplexing scheme is left for future research works.
\end{remark}
The estimate $\hat{\mathbf{H}}_{k,t_i^p,n_i^p}$ of $\mathbf{H}_{k,t_i^p,n_i^p}$ can be obtained from $\mathbf{p}_{k}\stackrel{\mathrm{def.}}{=}[p_{k,t_1,n_1} \cdots p_{k,t_{N_k^p},n_{N_k^p}}]^{\mathrm{T}}$ using least squares (LS) or linear minimum mean squared error (LMMSE) estimation as
\begin{equation}
\label{eq:p_est}
\hat{\mathbf{H}}_{k,t_i^p,n_i^p}=\mathbf{H}_{k,t_i^p,n_i^p}+
\frac{1}{\sqrt{\gamma_k^{\mathrm{CE}}}}
\mathbf{n}_{k,t_i^p,n_i^p},
\end{equation}
where $\mathbf{n}_{k,t_i^p,n_i^p}\sim\mathcal{CN}\left(\mathbf{0},\mathbf{I}_{M}\right)$ is an estimation error vector that is uncorrelated with $\mathbf{H}_{k,t_i^p,n_i^p}$ and where $\gamma_k^{\mathrm{CE}}$ is the effective channel estimation signal-to-noise ratio (SNR) given by  $\gamma_k^{\mathrm{CE}}=\frac{g_k P_c}{\sigma^2}$. Define $\hat{\mathbf{H}}_k^p$ as the $N_p M\times1$ vector made by concatenating the $N_k^p$ channel estimates $\hat{\mathbf{H}}_{k,t_i^p,n_i^p}$:
\begin{equation}
\label{eq:hat_H_p}
\hat{\mathbf{H}}_k^p\stackrel{\mathrm{def.}}{=}
\left[\hat{\mathbf{H}}_{k,t_1^p,n_1^p}^{\mathrm{T}}\:\cdots\:
\hat{\mathbf{H}}_{k,t_{N_k^p}^p,n_{N_k^p}^p}^{\mathrm{T}}\right]^{\mathrm{T}}
\in\mathbb{C}^{N_p M\times1}.
\end{equation}
Now, define $\mathbf{H}_k$ as the vector made by concatenating the $N$ channel vectors $\mathbf{H}_{k,t_i,n_i}$ at the position $(t_i,n_i)\in\mathcal{N}$ to which the spread samples of data symbol $s_k$ are mapped, i.e.,
\begin{equation}
\label{eq:H_k}
\mathbf{H}_k\stackrel{\mathrm{def.}}{=}
\left[\mathbf{H}_{k,t_1,n_1}^{\mathrm{T}}\:\cdots\:
\mathbf{H}_{k,t_N,n_N}^{\mathrm{T}}\right]^{\mathrm{T}}
\in\mathbb{C}^{N M\times1}.
\end{equation}
We propose to estimate $\mathbf{H}_k$ using LMMSE interpolation based on $\hat{\mathbf{H}}_k^p$. The resulting estimate (denoted $\hat{\mathbf{H}}_k$) is given by
\begin{equation}
\label{eq:hat_H_k}
\hat{\mathbf{H}}_k=\mathbf{R}_k^{\mathcal{NP}}\mathbf{Q}_k^{\mathcal{P}}
\hat{\mathbf{H}}_k^p\in\mathbb{C}^{N M\times1},
\end{equation}
with $\mathbf{R}_k^{\mathcal{NP}}=\mathbb{E}\left[\mathbf{H}_k\left(\hat{\mathbf{H}}_k^p\right)^{\mathrm{H}}\right]$,
 $\mathbf{Q}_k^{\mathcal{P}}=\left(\frac{1}{\gamma_k^{\mathrm{CE}}}\mathbf{I}_{MN_k^p}+
\mathbf{R}_k^{\mathcal{P}}\right)^{-1}$
and $\mathbf{R}_k^{\mathcal{P}}=\mathbb{E}\left[\hat{\mathbf{H}}_k^p
\left(\hat{\mathbf{H}}_k^p\right)^{\mathrm{H}}\right]$. Using the basic properties of LMMSE estimation, it is straightforward to show that $\mathbf{H}_k=\hat{\mathbf{H}}_k+\tilde{\mathbf{H}}_k$, where $\tilde{\mathbf{H}}_k\sim\mathcal{CN}\left(\mathbf{0}, 	\mathbf{R}_k-\boldsymbol{\Phi}_k\right)$ and
\begin{align}
 \label{eq:MatPhi}
&\boldsymbol{\Phi}_k\stackrel{\mathrm{def.}}{=}\mathbf{R}_k^{\mathcal{NP}}\mathbf{Q}_k^{\mathcal{P}}
\left(\mathbf{R}_k^{\mathcal{NP}}\right)^{\mathrm{H}},\\
\label{eq:MatRN}
&\mathbf{R}_k\stackrel{\mathrm{def.}}{=}\left[
\sum_{l=0}^{L-1}r_k^{\alpha}(t_i-t_j)\mathbf{R}_{k,l}^{\alpha}
e^{-2\pi\imath\frac{\tau_l (n_i-n_j)}{N_{\mathrm{FFT}} T_s}}
\right]_{1\leq i,j\leq N}.
\end{align}

\subsubsection{Data Detection}

It is useful to define
\begin{equation}
 \label{eq:Mat_C_k}
\mathbf{C}_k\stackrel{\mathrm{def.}}{=}
\mathrm{diag}\left(\mathbf{c}_k\right)\otimes\mathbf{I}_M\:.
\end{equation}
Data detection is done using $\sqrt{g_k P_c}\mathbf{C}_k\hat{\mathbf{H}}_k$ (the effective $N M$-long channel vector when the spreading code is incorporated) as the signature of UT $k$. Moreover, we assume that the BS employs minimum mean squared error (MMSE) detection for the higher-rate service classes and matched-filtering (MF) for the lower-rate ones using the following vector
\begin{equation}
 \label{eq:detection_r}
\mathbf{r}_k=\left\{
\begin{array}{ll}
\mathbf{r}_k^{\mathrm{MF}},& \textrm{MF detection,}\\
\mathbf{r}_k^{\mathrm{MMSE}},& \textrm{MMSE detection},
\end{array}
\right.
\end{equation}
where $\mathbf{r}_k^{\mathrm{MF}}=\sqrt{g_k P_c}\mathbf{C}_k\hat{\mathbf{H}}_k$ and $\mathbf{r}_k^{\mathrm{MMSE}}=\left(\sum_{j\in\mathcal{K}_c}g_jP_{c}\mathbf{C}_j\hat{\mathbf{H}}_j\hat{\mathbf{H}}_j^{\mathrm{H}}\mathbf{C}_j^{\mathrm{H}}+\sigma^2\mathbf{I}_{NM}\right)^{-1}\sqrt{g_k P_c}\mathbf{C}_k\hat{\mathbf{H}}_k$.

\section{Achievable Rates and Asymptotic Analysis}
\label{sec:perf}

The ergodic rate $R_k$ in bits/s/Hz achieved in the uplink by UT~$k$ from the $c$-th service class is given by $R_k=\mathbb{E}\left[\log_2\left(1+\gamma_k\right)\right]$ where the expectation is taken with respect to the joint distribution of $\{\hat{\mathbf{H}}_k,\mathbf{H}_k\}_{k\in\{1;\ldots,K\}}$ and where $\gamma_k$ is the signal-to-interference-plus-noise ratio given by \eqref{eq:class_c_SINR}.
\begin{figure*}[!t]
\normalsize
\begin{align}
 \label{eq:class_c_SINR}
\gamma_k=&
\frac{P_c g_k\left|\mathbf{r}_k^{\mathrm{H}}
\mathbf{C}_k\hat{\mathbf{H}}_k\right|^2}
{\mathbf{r}_k^{\mathrm{H}}
\left(\sigma^2\mathbf{I}_{MN}+
P_c g_k\mathbf{C}_k\tilde{\mathbf{H}}_k
\tilde{\mathbf{H}}_k^{\mathrm{H}}\mathbf{C}_k^{\mathrm{H}}-P_c g_k
\mathbf{C}_k\mathbf{H}_k\mathbf{H}_k^{\mathrm{H}}\mathbf{C}_k^{\mathrm{H}}+
\sum_{\tilde{c}=1}^{C}\sum_{j\in\mathcal{K}_{\tilde{c}}}P_{\tilde{c}}g_j
\mathbf{C}_j\mathbf{H}_j\mathbf{H}_j^{\mathrm{H}}\mathbf{C}_j^{\mathrm{H}}
\right)\mathbf{r}_k},\quad
\forall k\in\mathcal{K}_c\\
\label{eq:class_c_barSINR_MF}
\bar{\gamma}_k^{\mathrm{MF}}=&
\frac{P_c g_k\left(
	\frac{1}{NM}\mathrm{tr}
	\mathbf{C}_k\boldsymbol{\Phi}_k\mathbf{C}_k^{\mathrm{H}}
	\right)^2}{\frac{\sigma^2}{(NM)^2}\mathrm{tr}
	\mathbf{C}_k\boldsymbol{\Phi}_k\mathbf{C}_k^{\mathrm{H}}+
	\frac{1}{NM}\sum_{\tilde{c}=1}^{C}\sum_{j\in\mathcal{K}_{\tilde{c}}}
	P_{\tilde{c}}g_j\frac{1}{NM}\mathrm{tr}
	\mathbf{C}_j\mathbf{R}_j\mathbf{C}_j^{\mathrm{H}}
	\mathbf{C}_k\boldsymbol{\Phi}_k\mathbf{C}_k^{\mathrm{H}}},\quad
\forall k\in\mathcal{K}_c
\end{align}
\end{figure*}
To derive asymptotic deterministic equivalents of the SINR in \eqref{eq:class_c_SINR}, the following assumption is needed.
\begin{assumption}
 \label{ass:asymptotic}
For all $k\in\{1,\ldots,K\}$,$l\in\{0,\ldots,L-1\}$:
\begin{equation} 
\limsup_M\|\mathbf{R}_{k,l}^{\alpha}\|<\infty,\quad
\liminf_M\frac{1}{M}\mathrm{tr}\mathbf{R}_{k,l}^{\alpha}>0.
\end{equation}
\end{assumption}
\begin{theo}
 \label{theo:asymp_rates_MF}
If $\mathbf{r}_k=\mathbf{r}_k^{\mathrm{MF}}$, then under Assumption~\ref{ass:asymptotic} $\gamma_k-\bar{\gamma}_k^{\mathrm{MF}}
\xrightarrow[M,K_c\to\infty,\lim K_c/M<\infty]{a. s.} 0$ where $\bar{\gamma}_k^{\mathrm{MF}}$ is defined in \eqref{eq:class_c_barSINR_MF}.
\end{theo}
\begin{proof}[Sketch of the proof]
The proof follows from \cite[Theorem 3]{ul_dl} if matrices $\mathbf{S}_k=g_k P_c\mathbf{C}_k\mathbf{R}_k\mathbf{C}_k^{\mathrm{H}}$ and $\mathbf{F}_k=g_k P_c\mathbf{C}_k\boldsymbol{\Phi}_k\mathbf{C}_k^{\mathrm{H}}$ satisfy $\limsup_M\|\mathbf{S}_k\|<\infty$, $\liminf_M\frac{1}{NM}\mathrm{tr}\mathbf{S}_k>0$ and $\limsup_M\left(\frac{1}{K}\sum_{j\in\{1,\ldots,K\}} \frac{1}{NM}\mathrm{tr}\mathbf{F}_j\right)^{-1}<\infty$. These conditions can be shown to hold under Assumptions~\ref{ass:asymptotic} by using basic properties of the trace function and of the spectral norm.
\end{proof}
Now, let $\left\{\{\mathbf{S}_j\}_{j\in\{1\cdots J\}},\mathbf{F}\right\}$ be a set of $J+1$ Hermitian non-negative definite $I\times I$ matrices, each with a uniformly bounded spectral norm with respect to $I$. We define for any $\rho>0$ the matrix $\mathbf{T}\left(\rho,I,\{\mathbf{S}_j\}_{j\in\{1\cdots J\}}\right)\in\mathbb{C}^{I\times I}$ as 
\begin{equation}
 \label{eq:T_def}
\mathbf{T}\stackrel{\mathrm{def.}}{=}\left(\frac{1}{I}\sum_{j\in\{1\cdots J\}}\frac{\mathbf{S}_j}{1+\delta_j}+\rho\mathbf{I}_I\right)^{-1}
\end{equation}
where for all $k\in\{1,\ldots,J\}$ and $t\geq1$
\begin{align}
 \label{eq:delta_lim}
\delta_k\left(\rho,I,\{\mathbf{S}_j\}_{j\in\{1\cdots J\}}\right)
\stackrel{\mathrm{def.}}{=}\lim_{t\to\infty}\delta_k^{(t)},\\
 \label{eq:delta_t}
\delta_k^{(t)}\stackrel{\mathrm{def.}}{=}\frac{1}{I}\mathrm{tr}\mathbf{S}_k
\left(\frac{1}{I}\sum_{j\in\{1\cdots J\}}\frac{\mathbf{S}_j}{1+\delta_j^{(t-1)}}+
\rho\mathbf{I}_I\right)^{-1}.
\end{align}
If we set $\delta_k^{(0)}=\frac{1}{\rho}$ for all $k\in\{1,\ldots,J\}$, then the limit in~\eqref{eq:delta_lim} exists due to~\cite[Theorem~1]{couillet}. We also define
\begin{equation}
 \label{eq:T_prime_def}
\mathbf{T}'\stackrel{\mathrm{def.}}{=}
\mathbf{T}\mathbf{F}\mathbf{T}+\frac{1}{I}\mathbf{T}\sum_{j\in\{1\cdots J\}}\frac{\mathbf{S}_j\delta_j^{'}}{\left(1+\delta_j\right)^2}\mathbf{T},
\end{equation}
where $\{\delta_j^{'}\left(\rho,I,\{\mathbf{S}_j\}_{j\in\{1\cdots J\}},
\mathbf{F}\right)\}_{j\in\{1\cdots J\}}$ are the entries of the
$J$-long vector $\boldsymbol{\delta}'$ defined as $\boldsymbol{\delta}'\stackrel{\mathrm{def.}}{=}\left(\mathbf{I}_J-\mathbf{J}\right)^{-1}\mathbf{v}$
with
\begin{equation}
\left[\mathbf{J}\right]_{kj}\stackrel{\mathrm{def.}}{=}\frac{\frac{1}{I}\mathrm{tr}\mathbf{S}_k\mathbf{T}\mathbf{S}_j\mathbf{T}}{I\left(1+\delta_j\right)^2},\qquad
\left[\mathbf{v}\right]_k\stackrel{\mathrm{def.}}{=}
\frac{1}{I}\mathrm{tr}\mathbf{S}_k\mathbf{T}\mathbf{F}\mathbf{T}
\end{equation}
for any $(k,j)\in\{1,\ldots,J\}^2$.
MOMA's asymptotic ergodic rates for classes with MMSE detection can be characterized with the help of the above definitions by the following theorem. Its proof is omitted here due to page count limitations.
\begin{theo}
 \label{theo:asymp_rates_MMSE}
If $\mathbf{r}_k=\mathbf{r}_k^{\mathrm{MMSE}}$, then under Assumption~\ref{ass:asymptotic} $\gamma_k-\bar{\gamma}_k^{\mathrm{MMSE}}\xrightarrow[M,K_c\to\infty,\lim K_c/M<\infty]{a. s.} 0$ where $\forall k\in\mathcal{K}_c$
\begin{equation}
 \label{eq:class_c_barSINR_MMSE}
 \begin{multlined}
\bar{\gamma}_k^{\mathrm{MMSE}}=\\
\frac{\delta_{ck}^2}{\frac{\sigma^2}{(NM)^2}\mathrm{tr}P_c g_k
\mathbf{C}_k\boldsymbol{\Phi}_k\mathbf{C}_k^{\mathrm{H}}
\bar{\mathbf{T}}_{c}^{'}+
\frac{1}{NM}\sum_{\tilde{c}=1}^{C}
\sum_{j\in\mathcal{K}_{\tilde{c}}}\mu_{c\tilde{c}kj}},
 \end{multlined}
\end{equation}
and
\begin{equation}
\label{eq:mu_mmse_def}
 \begin{multlined}
  \mu_{c\tilde{c}kj}\stackrel{\mathrm{def.}}{=}
  \frac{P_{\tilde{c}} g_j}{NM}\mathrm{tr}
  \mathbf{C}_j\mathbf{R}_j\mathbf{C}_j^{\mathrm{H}}\mathbf{T}_{ck}^{'}-\\
  \mathbbm{1}_{\mathcal{K}_c}(j)\frac{2\mathrm{Re}\left(
  	\vartheta_{c\tilde{c}j}^*\vartheta_{c\tilde{c}kj}^{'}\right)(1+\delta_{cj})-
  	|\vartheta_{c\tilde{c}j}|^2\delta_{ckj}^{'}}{(1+\delta_{cj})^2},
 \end{multlined}
\end{equation}
\begin{align}
\vartheta_{c\tilde{c}j}\stackrel{\mathrm{def.}}{=}&\frac{P_{\tilde{c}} g_j}{NM}\mathrm{tr}
\mathbf{C}_j\boldsymbol{\Phi}_j\mathbf{C}_j^{\mathrm{H}}\mathbf{T}_c,\\
\vartheta_{c\tilde{c}kj}^{'}\stackrel{\mathrm{def.}}{=}&\frac{P_{\tilde{c}} g_j}{NM}\mathrm{tr}
\mathbf{C}_j\boldsymbol{\Phi}_j\mathbf{C}_j^{\mathrm{H}}\mathbf{T}_{ck}^{'}.
\end{align}
Here, we defined $\delta_{ck}$ $=\delta_k\left(\frac{\sigma^2}{M},NM,
\{P_c g_j\mathbf{C}_j\boldsymbol{\Phi}_j\mathbf{C}_j^{\mathrm{H}}\}_{j\in\mathcal{K}_c}\right)$, 
$\delta_{ck}^{'}=\delta_{k}^{'}
\left(\frac{\sigma^2}{M},NM,\{P_c g_j\mathbf{C}_j\boldsymbol{\Phi}_j\mathbf{C}_j^{\mathrm{H}}\}_{j\in\mathcal{K}_c},
\boldsymbol{\Phi}_k\right)$,
$\mathbf{T}_c=\mathbf{T}\left(\frac{\sigma^2}{M},NM,
\{P_c g_j\mathbf{C}_j\boldsymbol{\Phi}_j\mathbf{C}_j^{\mathrm{H}}\}_{j\in\mathcal{K}_c}\right)$,
$\mathbf{T}_{ck}^{'}=
\mathbf{T}'\left(\frac{\sigma^2}{M},NM,\{P_c g_j\mathbf{C}_j\boldsymbol{\Phi}_j\mathbf{C}_j^{\mathrm{H}}\}_{j\in\mathcal{K}_c},
\boldsymbol{\Phi}_k\right)$, and
$\bar{\mathbf{T}}_c^{'}=\mathbf{T}'\left(
\frac{\sigma^2}{M},NM,\{P_c g_j\mathbf{C}_j\boldsymbol{\Phi}_j\mathbf{C}_j^{\mathrm{H}}\}_{j\in\mathcal{K}_c},
\mathbf{I}_{NM}\right)$.
\end{theo}
Let us now consider the case where the time and frequency domain correlation functions are flat over the RE set $\mathcal{N}$, i.e., $r_k^{\alpha}(t_i-t_j)=1$ and $e^{-2\pi\imath\frac{\tau_l (n_i-n_j)}{N_{\mathrm{FFT}} T_s}}=1$. In this case, $\mathbf{R}_k=\mathbf{1}_{N\times N}\otimes\check{\mathbf{R}}_k$,  $\mathbf{R}_k^{\mathcal{NP}}=\mathbf{1}_{N\times1}\otimes \left[\check{\mathbf{R}}_{k,1}^{\mathcal{NP}}\:\cdots\: \check{\mathbf{R}}_{k,N_p}^{\mathcal{NP}}\right]$, $\boldsymbol{\Phi}_k=\mathbf{1}_{N\times N}\otimes\check{\boldsymbol{\Phi}}_k$, where $\forall j\in\{1,\ldots,N_p\}$
\begin{align}
 \label{eq:small_N_MatRN_check}
\check{\mathbf{R}}_k&\stackrel{\mathrm{def.}}{=}
\sum_{l=0}^{L-1}\mathbf{R}_{k,l}^{\alpha}\\
\check{\mathbf{R}}_{k,j}^{\mathcal{NP}}&\stackrel{\mathrm{def.}}{=}
\sum_{l=0}^{L-1}r_k^{\alpha}(t_1-t_j^p)\mathbf{R}_{k,l}^{\alpha}
e^{-2\pi\imath\frac{\tau_l (n_1-n_j^p)}{N_{\mathrm{FFT}} T_s}},\\
\check{\boldsymbol{\Phi}}_k&\stackrel{\mathrm{def.}}{=}
\sum_{i=1}^{N_p}\sum_{j=1}^{N_p}\check{\mathbf{R}}_{k,i}^{\mathcal{NP}}
\left[\mathbf{Q}_k^{\mathcal{P}}\right]_{i,j}
\left(\check{\mathbf{R}}_{k,j}^{\mathcal{NP}}\right)^{\mathrm{H}}
\end{align}
The following assumption is a relaxation of strict equality to 1 of the time and frequency domain correlation coefficients.
\begin{assumption}
 \label{ass:small_N}
For all $k\in\{1,\ldots,K\}$:
\begin{align*}
\mathbf{R}_k&=\mathbf{1}_{N\times N}\otimes\check{\mathbf{R}}_k+
\tilde{\mathbf{R}}_k,\\
\mathbf{R}_k^{\mathcal{NP}}&=\mathbf{1}_{N\times1}\otimes
\left[\check{\mathbf{R}}_{k,1}^{\mathcal{NP}}\:\cdots\:
\check{\mathbf{R}}_{k,N_p}^{\mathcal{NP}}\right]+
\tilde{\mathbf{R}}_k^{\mathcal{NP}},
\end{align*}
with $\tilde{\mathbf{R}}_k\in\mathbb{C}^{NM\times NM}$ and
$\tilde{\mathbf{R}}_k^{\mathcal{NP}}\in\mathbb{C}^{NM\times N_pM}$ satisfying
$\lim_{M\to\infty}\left\|\tilde{\mathbf{R}}_k\right\|=
\lim_{M\to\infty}\left\|\tilde{\mathbf{R}}_k^{\mathcal{NP}}\right\|=0$.
\end{assumption}
Note that in practice, if we choose $\mathcal{N}$ to consist of at most few tens of REs, Assumption~\ref{ass:small_N} becomes a mild one.
\begin{corollary}
 \label{cor:small_N_MF}
Under Assumptions~\ref{ass:asymptotic} and \ref{ass:small_N}, we have $\forall k\in\mathcal{K}_c$
\begin{equation}
 \label{eq:small_N_class_c_barSINR_MF}
\bar{\gamma}_k^{\mathrm{MF}}=\frac{P_c g_k\left(
\frac{1}{M}\mathrm{tr}\check{\boldsymbol{\Phi}}_k\right)^2}{
\frac{\sigma^2}{M^2}\mathrm{tr}\check{\boldsymbol{\Phi}}_k+
\frac{1}{M}\sum_{j\in\mathcal{K}_{c}}
P_{c}g_j\frac{1}{M}\left(\mathbf{w}_j^{\mathrm{H}}\mathbf{w}_k\right)^2
\mathrm{tr}\check{\mathbf{R}}_j\check{\boldsymbol{\Phi}}_k}.
\end{equation}
\end{corollary}
\begin{proof}[Sketch of the proof] Note that under Assumption~\ref{ass:small_N} and due to \eqref{eq:MatPhi}, $\boldsymbol{\Phi}_k=\mathbf{1}_{N\times N}\otimes\check{\boldsymbol{\Phi}}_k+\tilde{\boldsymbol{\Phi}}_k$ where $\lim_{M\to\infty}\left\|\tilde{\boldsymbol{\Phi}}_k\right\|=0$ for all $k\in\{1,\ldots,K\}$. Also note that for any pair $(k,j)\in\mathcal{K}_c^2$, $\mathbf{c}_j^{\mathrm{H}}\mathbf{c}_k=	\mathbf{w}_j^{\mathrm{H}}\mathbf{w}_k$ due to \eqref{eq:spreading_code_k_c}. Finally, use basic properties of the trace function and of the Kronecker product to conclude, after some tedious derivations, that \eqref{eq:small_N_class_c_barSINR_MF} holds. 
\end{proof}
\begin{corollary}
 \label{cor:small_N_MMSE}
Under Assumptions~\ref{ass:asymptotic} and \ref{ass:small_N}, we have $\forall k\in\mathcal{K}_c$
\begin{equation}
 \label{eq:small_N_class_c_barSINR_MMSE}
\bar{\gamma}_k^{\mathrm{MMSE}}=
\frac{\delta_{ck}^2}{
\frac{\sigma^2}{(NM)^2}\mathrm{tr}P_c g_k
\mathbf{C}_k\boldsymbol{\Phi}_k\mathbf{C}_k^{\mathrm{H}}
\bar{\mathbf{T}}_{c}^{'}+
\frac{1}{NM}\sum_{j\in\mathcal{K}_c}\mu_{cckj}}.
\end{equation}
\end{corollary}
\begin{proof}[Sketch of the proof] First, note that for any $j\in\mathcal{K}_{\tilde{c}}$ with $\tilde{c}\neq c$ the term $\mu_{c\tilde{c}kj}$ in \eqref{eq:class_c_barSINR_MMSE} can be written due to \eqref{eq:mu_mmse_def} as
\begin{equation}
 \begin{multlined}
 \label{eq:mu_j_tilde_c}
\mu_{c\tilde{c}kj}=\frac{P_{\tilde{c}} g_j}{NM}\mathrm{tr}\mathbf{C}_j\mathbf{R}_j\mathbf{C}_j^{\mathrm{H}}\mathbf{T}_c	\mathbf{C}_k\boldsymbol{\Phi}_k\mathbf{C}_k^{\mathrm{H}}\mathbf{T}_c+\\
\frac{P_{\tilde{c}} g_j}{(NM)^2}\sum_{i\in\mathcal{K}_c}\frac{\delta_{ci}^{'}}{\left(1+\delta_{ci}\right)^2}\mathrm{tr}\mathbf{C}_j\mathbf{R}_j\mathbf{C}_j^{\mathrm{H}}\mathbf{T}_c\mathbf{C}_i\boldsymbol{\Phi}_i\mathbf{C}_i^{\mathrm{H}}\mathbf{T}_c\:.
 \end{multlined}
\end{equation}
Second, refer to \eqref{eq:T_def} and use the Cayley-Hamilton theorem~\cite{matrix_analysis} to write $\mathbf{T}_c$ as a weighted sum of the first $NM$ powers of the matrix $\frac{1}{NM}\sum_{j\in\mathcal{K}_c}\frac{\boldsymbol{\Phi}_j}{1+\delta_{cj}}+\frac{\sigma^2}{M}\mathbf{I}_{NM}$. Finally, plug this weighted sum into \eqref{eq:mu_j_tilde_c} and use Assumption~\ref{ass:small_N} along with the fact that $\mathbf{c}_j^{\mathrm{H}}\mathbf{c}_k=0$ for any $(k,j)\in\mathcal{K}_c\times\mathcal{K}_{\tilde{c}}$ satisfying $\tilde{c}\neq c$ to conclude after some tedious calculation that  $\mu_{c\tilde{c}kj}\xrightarrow[M\to\infty]{a. s.}0$ for any $j\in\mathcal{K}_{\tilde{c}}$, thus completing the proof of Corollary~\ref{cor:small_N_MMSE}.
\end{proof}
Asymptotic inter-class orthogonality is thus achieved by MOMA without requiring the actual channel gains over the different $N$ REs used for spreading one data symbol to be equal. This relaxation can be thought of as a consequence of the \emph{massive MIMO effect} which restores the orthogonality of the original $N\times N$ code matrix. Now, to get more insight into the intra-class multiuser interference performance achieved by MOMA, we focus on MF detection and we consider a simplified channel model by assuming
\begin{equation}
 \label{eq:uncorr_spatial}
\mathbf{R}_{k,l}^{\alpha}=\sigma_l^2\mathbf{I}_M,\:\forall k\in\{1\ldots,K\},
l\in\{0,\ldots,L-1\}.
\end{equation}
\begin{corollary}
 \label{cor:simplified_MF}
Assume that the empirical distribution of the large-scale fading coefficients $\{g_k\}_{k\in\mathcal{K}_c}$ converges as $K_c\to\infty$ to the distribution of a random variable with mean $\bar{g}$. If the columns of $\mathbf{W}_c$ are further chosen as realizations of i.i.d. zero-mean random vectors with independent entries that satisfy the average power constraint in \eqref{eq:avg_tx_p_constraint}, then under the simplified channel model in \eqref{eq:uncorr_spatial} and Assumptions~\ref{ass:asymptotic} and \ref{ass:small_N} the deterministic equivalent $\bar{\gamma}_k^{\mathrm{MF}}$ defined in \eqref{eq:class_c_barSINR_MF} is given as
\begin{equation}
 \label{eq:simplified_barSINR_MF}
\bar{\gamma}_k^{\mathrm{MF}}=\frac{P_c g_k
\frac{\mathrm{tr}\check{\boldsymbol{\Phi}}_k}{M}}{
\frac{\sigma^2}{M}+\frac{\bar{g}P_c K_c}{N_c M}},
\forall k\in\mathcal{K}_c\:.
\end{equation}
\end{corollary}
Note that the assumption on the empirical distribution of $\{g_k\}_{k\in\mathcal{K}_c}$ is a mild technical assumption which holds, for instance, if the positions of the UTs of each service class are scattered diversely enough within the cell area.
Associating massive MIMO and MOMA thus amounts to an effective spreading gain equal to $N_c M$, thus validating \eqref{eq:scenario_of_interest}.
\begin{proof}[Sketch of the proof]
Plugging \eqref{eq:uncorr_spatial} into \eqref{eq:small_N_MatRN_check} gives $\check{\mathbf{R}}_k=\sum_{l=0}^{L-1}\sigma_l^2\mathbf{I}_M=\mathbf{I}_M$. Since the entries of $\mathbf{W}_c$ are zero-mean i.i.d., then $\frac{1}{M}\sum_{j\in\mathcal{K}_{c}}g_j P_c \left(\mathbf{w}_j^{\mathrm{H}}\mathbf{w}_k\right)^2 \mathrm{tr}\check{\mathbf{R}}_j\check{\boldsymbol{\Phi}}_k$ in \eqref{eq:small_N_class_c_barSINR_MF} is asymptotically equivalent to $\mathrm{tr}\check{\boldsymbol{\Phi}}_k\frac{\bar{g}P_c K_c}{N_c M}$ as $K_c\to\infty$ due to the \emph{central limit theorem} and to the assumption made about the empirical distribution of $\{g_k\}_{k\in\mathcal{K}_c}$. Putting things together completes the proof that \eqref{eq:simplified_barSINR_MF} holds.
\end{proof}


\section{Numerical Results}
\label{sec:simu}

In this section, we consider a LTE-like system with a total bandwidth $W=10$ MHz out of which MOMA occupies a sub-band covering one RB. On this RB, TTI bundling is applied on the basis of 6 consecutive TTIs i.e., $N=6$. We also assume $C=2$, i.e., the terminals scheduled on the considered resource region belong to two service classes. The noise power spectral density is equal to $N_0=-174$ dBm/Hz and users' channels follow the Extended Type Urban (ETU) model \cite{3gpp_etu_epa} with $f_k^D=70$ Hz and $\max_l\tau_l=5$ $\mu$s. Channel vector spatial correlation is simulated using the {\it physical channel model}~\cite{ul_dl} for uniform linear arrays with a number of dimensions equal to half the number of BS antennas. The number $N_p/N$ of pilot symbols per RB is equal to the total number of UTs, i.e., $N_p=N(K_1+K_2)$. Devices of the first class emit with power $P_1=23$ dBm while the transmit power of the second class is $P_2=17$ dBm. The SNR $P_c g_k/\sigma^2$ of the signals received from the terminals of the first class is equal to $15$ dB while the SNR of the second class is equal to 1.5 dB. This difference in SNR values is intended to simulate the effects of different transmit power levels and different propagation conditions that differentiate in practice among the devices from different service classes. All the following results have been obtained by averaging over 100 realizations of the wireless channels.

The orthogonal code matrix $\mathbf{U}$ is a $6\times6$ DFT matrix and $N_1=N_2=3$. The columns of matrices $\mathbf{W}_1$ and $\mathbf{W}_2$ are random vectors with i.i.d. modulus-one complex entries. In the following, we compare MOMA to three multiple-access schemes implemented on the same resource region, namely ``MIMO-eMTC with repetitions'', ``MIMO-eMTC with spreading'', and NB-IoT. By ``MIMO-eMTC with repetitions'' we mean an eMTC system where several UTs are allowed to be scheduled on the same RB (hence the MIMO prefix) while there signals are repeated on 6 consecutive TTIs. In ``MIMO-eMTC with spreading'', signal repetition is replaced with random binary spreading (using the same resource mapping as in MOMA). Finally, NB-IoT designates a pure FDMA system where the available bandwidth is split into $K$ sub-channels, each assigned to a single UT, with inter-sub-channel guard bands that amount to 10\% of the bandwidth of the RB.

First, assume that $K_1=6$ and $K_2=18$. Figure~\ref{fig:r_vs_M} shows the maximum ergodic rate that can be achieved for the two service classes under this assumption. The deterministic equivalent of this rate as given by Corollaries~\ref{cor:small_N_MF} and~\ref{cor:small_N_MMSE} is also plotted in the same figure and is shown to be close to the simulated value even for numbers of BS antennas as small as 16. The approximation is slightly optimistic due to the fact that residual inter-class interference at finite values of $M$ is not accounted for.
\begin{figure}
 \centering
 \includegraphics[width=1.00\hsize]{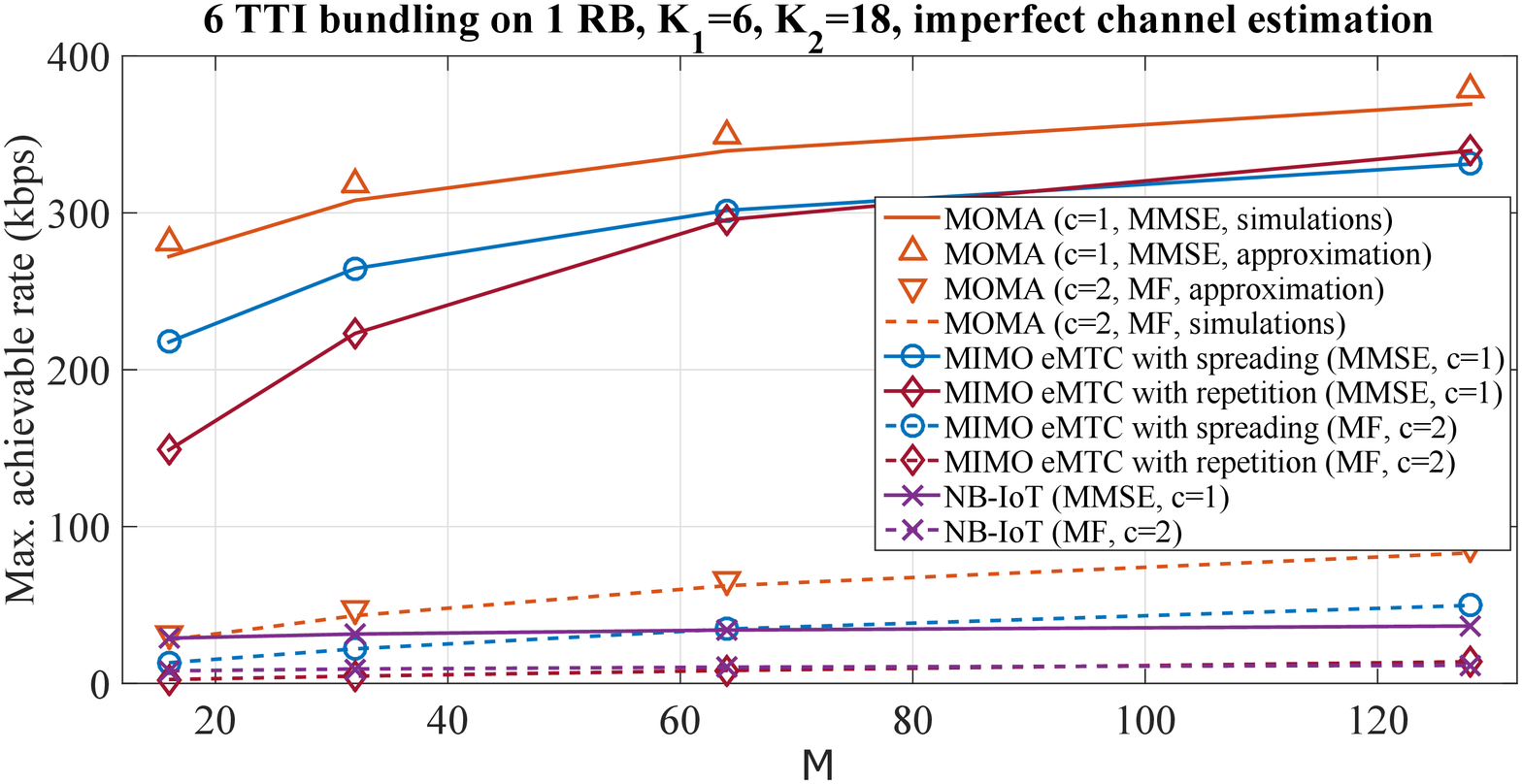}
 \caption{Maximum achievable data rate vs. BS array size}
 \label{fig:r_vs_M}
\end{figure}
Also note that the schemes using spreading, namely MOMA and ``MIMO-eMTC with spreading'', outperform as expected ``MIMO-eMTC with repetition''. Moreover, MOMA's service dependent hierarchical spreading achieves higher data rates for both service classes than conventional random spreading. This is thanks to inter-class quasi-orthogonality achieved by the former spreading technique. The smaller bandwidth per transmission in NB-IoT limits the achievable data rate, especially for the UTs of the first class.

Now, assume that $M=64$ and $K_1=3$. Figure~\ref{fig:K_vs_r} shows in this case the maximum value of $K_2$ that can be supported while $R_k\geq r_2$ for all $k\in\mathcal{K}_2$. Again, MOMA outperforms all the other multiple-access schemes on most of the considered range of target rate values.
\begin{figure}
\centering
\includegraphics[width=1.00\hsize]{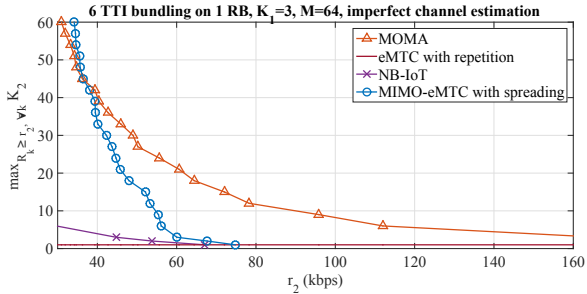}
\caption{Maximum number of simultaneous transmissions from the second class vs. target data rate}
\label{fig:K_vs_r}
\end{figure}
It is worth pointing out that the bundling frame structure needs not be aligned in order for MOMA's spreading to provide the multiplexing gain plotted in Figure~\ref{fig:K_vs_r}. In other words, if one device's first bundled sub-frame coincides with other devices' second (or third or \dots) bundled sub-frames, their respective MOMA's codewords will still allow the BS to distinguish them from each other. This is thanks to the particular resource mapping we propose for MOMA (see Subsection~\ref{sec:moma_mapping}) and which consists in mapping the spread samples to neighboring REs within one RB/sub-frame rather than to neighboring RBs or sub-frames.


\section{Conclusions}

MOMA can create high resource granularity on any sub-region of the resource grid without the need to change the numerology used in that region. This results in an intrinsic flexibility in resource assignment to different classes and concurrent transmissions within them. Moreover, using class dependent hierarchical spreading in MOMA instead of signal repetition, FDMA or conventional CDMA allows for more simultaneous transmissions to be multiplexed at higher target data rates, thus offering better system scalability with respect to increasing device densities. In this article, MOMA's scalability, efficiency and flexibility in resource utilization have been validated using both theoretical analysis and simulations.

\end{document}